\documentclass[a4paper, 12pt]{article}

\usepackage[utf8]{inputenc}
\usepackage[english]{babel}
\usepackage[T1]{fontenc}
\usepackage[letterpaper, margin=3cm]{geometry}
\usepackage{needspace}
\usepackage{bbm}

\usepackage{libertine} 
\usepackage{inconsolata} 

\usepackage[official]{eurosym}
\usepackage{setspace}

\usepackage{stackengine}

\usepackage{amsmath, amsthm, amssymb}
\usepackage{graphicx}
\usepackage{enumerate}
\usepackage{enumitem}
\usepackage{authblk}
\usepackage{thmtools}
\usepackage{thm-restate}
\usepackage[colorlinks=true, citecolor=red]{hyperref}
\usepackage[capitalize]{cleveref}
\usepackage{float}
\usepackage{amsfonts,mathtools,enumerate}
\usepackage[ruled,vlined]{algorithm2e}
\usepackage{url} 

\usepackage{tikz}
\usetikzlibrary{arrows,decorations,backgrounds,snakes,calc} 

\renewenvironment{abstract}
{\small\vspace{-1em}
\begin{center}
\bfseries\abstractname\vspace{-.5em}\vspace{0pt}
\end{center}
\list{}{
\setlength{\leftmargin}{0.6in}%
\setlength{\rightmargin}{\leftmargin}}%
\item\relax}
{\endlist}

\declaretheorem[name=Theorem, numberwithin=section]{theorem}
\declaretheorem[name=Lemma, sibling=theorem]{lemma}

\declaretheorem[name=Problem, sibling=theorem]{problem}
\declaretheorem[name=Claim, sibling=theorem]{claim}

\def\cqedsymbol{\ifmmode$\lrcorner$\else{\unskip\nobreak\hfil
\penalty50\hskip1em\null\nobreak\hfil$\lrcorner$
\parfillskip=0pt\finalhyphendemerits=0\endgraf}\fi}

\def\Pr{\mathbb{P}}
\def\Ex{\mathbb{E}}
\newcommand{\NN}{\mathbb{N}} 

\tikzstyle{vertex}=[circle, draw, fill=black!50,
                        inner sep=0pt, minimum width=4pt]

\let\leq\leqslant
\let\geq\geqslant

\newcommand{\qu}{\textsc{Query}}
\DeclareMathOperator{\tl}{tl}

\date{}

\title{Quasi-linear distance query reconstruction for graphs of bounded treelength}

\author[1]{Paul Bastide}
\author[2]{Carla Groenland}
\affil[1]{LaBRI - Université de Bordeaux, \href{mailto:paul.bastide@ens-rennes.fr}{paul.bastide@ens-rennes.fr}}
\affil[2]{TU Delft, \href{mailto:c.e.groenland@tudelft.nl}{c.e.groenland@tudelft.nl}}

\begin{document}

\maketitle
\begin{abstract}
In distance query reconstruction, we wish to reconstruct the edge set of a hidden graph by asking as few distance queries as possible to an oracle. Given two vertices $u$ and $v$, the oracle returns the shortest path distance between $u$ and $v$ in the graph.

The \emph{length} of a tree decomposition is the maximum distance between two vertices contained in the same bag. The \emph{treelength} of a graph is defined as the minimum length of a tree decomposition of this graph. We present an algorithm to reconstruct an $n$-vertex connected graph $G$ parameterized by maximum degree $\Delta$ and treelength $k$ in $O_{k,\Delta}(n \log^2 n)$ queries (in expectation). This is the first algorithm to achieve quasi-linear complexity for this class of graphs. The proof goes through a new lemma that could give independent insight on graphs of bounded treelength.
\end{abstract}

\section{Introduction}
There has been extensive study on identifying the structure of decentralized networks \cite{beerliova06,mathieu2013graph,rong21,mathieu2021simple,bastide2023optimal}. These networks are composed of vertices (representing servers or computers)   and edges (representing direct interconnections). To trace the path through these networks  from one actor to another, tools like \texttt{traceroute} (also known as \texttt{tracert}) were developed. 
If the entire route cannot be inferred (e.g. due to privacy concerns), a ping-pong protocol can be employed in which one node sends a dummy message to the second node, which then immediately responds with a dummy message back to the first node. This process aims to infer the distance between the nodes by measuring the time elapsed between the sending of the first message and the receipt of the second. 

A mathematical model for this called the \textit{distance query model} was introduced \cite{beerliova06}. In this model, only the vertex set \(V\) of a hidden graph \(G=(V,E)\) is known, and the goal is to reconstruct the edge set \(E\) through distance queries to an oracle. For any pair of vertices \((u,v) \in V^2\), the oracle provides the shortest path distance between \(u\) and \(v\) in \(G\). The algorithm can be adaptive and base its next query on the responses to previous queries. 

For a graph class \(\mathcal{G}\) of connected graphs, an algorithm is said to reconstruct the graphs in the class if, for every graph \(G \in \mathcal{G}\), the distance profile obtained is unique to \(G\) within \(\mathcal{G}\). We then say the graph has been reconstructed. The \textit{query complexity} refers to the maximum number of queries the algorithm executes on an input graph from \(\mathcal{G}\). For a randomised algorithm, the query complexity is determined by the expected number of queries, accounting for the algorithm's randomness. Such a randomised algorithm could also be seen as a probability distribution over decision trees. 

Note that querying the oracle for the distance between every pair of vertices in \(G\) would reconstruct the edge set as \(E = \{ \{u,v\} \mid d(u,v) = 1 \}\). This approach leads to a trivial upper bound of \(|V|^2\) on the query complexity. Unfortunately,  \(\Omega(|V|^2)\) queries may be required to, for example, distinguish between a clique \(K_n\) and \(K_n\) minus an edge \(\{u,v\}\). If the maximum degree is unbounded, this issue persists even in sparse graphs like trees: it can take $\Omega(n^2)$ queries to distinguish $n$-vertex trees (see also \cite{reyzin2007learning}). 
Therefore, as was also done in earlier work, we will restrict ourselves to connected $n$-vertex graphs with maximum degree $\Delta$. 

\paragraph{Previous work} Kannan, Mathieu and Zhou \cite{KannanMZ14,mathieu2013graph} were the first to give a non-trivial upper bound for all graphs of bounded maximum degree, designing a randomised algorithm using \(\Tilde{O}_{\Delta}(n^{3/2})\) queries in expectation for $n$-vertex graphs of maximum degree $\Delta$. Here  $\Tilde{O}(f(n))$ stands for $O(f(n)\operatorname{polylog}(n))$ and the $\Delta$ subscript denotes that $\Delta$ is considered a parameter and only influences the multiplicative constant in front of $f(n)$, (e.g here we mean $g(\Delta)n^{3/2} \operatorname{polylog} n$ for some function $g : \NN \mapsto \mathbb{R}$.). This is still the best known upper bound in the general case, while the best lower bound is $\Omega(\Delta n \log_\Delta n)$ \cite{bastide2023optimal}. Researchers spent effort investigating $\Tilde{O}_\Delta(n)$ algorithms for restricted classes of graphs. Kannan, Mathieu and Zhou \cite{KannanMZ14,mathieu2013graph} proved that there exists an $O_\Delta(n \log^3 n)$ randomised algorithm for chordal graphs (graphs without induced cycle of length at least 4).  
Since then, their algorithm for chordal graphs has been improved by Rong, Li, Yang, and Wang \cite{rong21} to \(O_{\Delta}(n\log^2 n)\), who also extended the class to $4$-chordal graphs (graphs without induced cycle of length at least $5$). Recent works introduced new techniques to design deterministic reconstruction algorithms \cite{bastide2023optimal}. They developed a quasi-linear algorithm for bounded maximum degree $k$-chordal graphs (without induced cycle of length at least $k+1$ and maximum degree $\Delta$) using $O_{\Delta,k}(n \log n)$ queries. Their results can be interpreted as a quasi-linear algorithm parameterized by maximum degree and chordality.
In this paper, we are the first to use a parameterized approach to extend on the techniques of Kannan, Mathieu and Zhou \cite{KannanMZ14,mathieu2013graph}, obtaining an algorithm with quasi-linear query complexity parametrized by even more general parameters.

\paragraph{Treelength} 
A graph $G$ has \emph{treelength} at most $\ell$ if it admits a tree decomposition such that $d_G(u,v)\leq \ell$ whenever $u,v\in V(G)$ share of a bag (see Section \ref{sec:prel} for formal definition). We emphasize that the bags are allowed to induce disconnected subgraphs, and that the `bounded diameter' constraint is measured within the entire graph.  Graphs of treelength $1$ are exactly chordal graphs and it was proved in \cite{kosowski2015k} that $k$-chordal graphs have treelength at most $k$. For $k>1$, the class of graphs of treelength at most $k$ covers a larger class of graphs than the class of $k$-chordal graphs.

Graphs of bounded treelength avoid long geodesic cycles (i.e. cycles $C$ for which $d_C(x,y)=d_G(x,y)$ for all $x,y\in C$) and in fact bounded treelength is equivalent to avoiding long `loaded geodesic cycles' or being `boundedly quasi-isometric to a tree' (see \cite{BergerSeymour23} for formal statements). When a graph has bounded treewidth (defined in \cref{sec:prel}), then the length of the longest geodesic cycle is bounded if and only if the  \emph{connected treewidth} is bounded \cite{diestel2018connected}. In a tree decomposition of connected treewidth at most $k$, bags induce connected subgraphs of size at most $k+1$, which in particular means that graph distance between vertices sharing a bag is at most $k$. So for graphs of bounded treewidth, excluding long geodesic cycles is in fact equivalent to bounding the treelength of the graph.

Treelength has been extensively studied from an algorithmic standpoint, particularly for problems related to shortest path distances. For example, there exist efficient routing schemes for graphs with bounded treelength \cite{dourisboure2007tree,kosowski2015k} and an FPT algorithm for computing the metric dimension of a graph parameterised by its treelength \cite{belmonte2017metric}. 
Although deciding the treelength of a given graph is NP-complete, it can still be approximated efficiently \cite{dourisboure2007tree,dissaux2021treelength}.

\paragraph{Our contribution}  Building on methods used by Kannan, Mathieu, and Zhou \cite{mathieu2013graph,KannanMZ14} to reconstruct chordal graphs, we prove the following result.
\begin{restatable}{theorem}{randtl}
\label{thm:randtl}
There is a randomised algorithm that reconstructs an $n$-vertex graph of maximum degree at most $\Delta$ and treelength at most $k$ using $O_{\Delta,k}(n \log^2 n)$ distance queries in expectation. 
\end{restatable}
We now first describe the technique used by Kannan, Mathieu and Zhou \cite{mathieu2013graph,KannanMZ14} for chordal graphs and then discuss our extension. In their approach, they design a clever subroutine to compute a small balanced separator $S$ of the graph $G$ using $\Tilde{O}_\Delta(n)$ queries. With the knowledge of this separator, it is possible to compute the partition in connected component of $G \setminus S$. By using this subroutine recursively, they are able to decompose the graph into smaller and smaller components until a brute-force search already yields a $\Tilde{O}_\Delta(n)$ queries algorithm. They exploit the strong structure of chordal graphs in two ways in this algorithm. First, to compute a small separator $S$. They start by only finding a single vertex that lies on many shortest paths. They then use a specific tree decomposition of chordal graphs, where all bags are cliques, to argue that the neighbourhood of this vertex is a good separator. Second, they show that for any connected component $C$ of $G \setminus S$ the distance between vertices in $C$ are the same in $G[C \cup S]$\footnote{Given a graph $G$ and a set of vertices $S \subseteq V(G)$, we use the notation $G[S]$ to denote the graphs induces by $G$ on the vertex set $S$.} and in $G$. This property allows to apply their subroutine recursively, as we can now simulate a distance oracle in $G[C \cap S]$ by just using the one we have on $G$. 

\cref{thm:randtl} shows that we can push the boundaries of such an approach, and proves that a weaker condition on the tree decomposition is already sufficient. We weaken the \emph{`bags are cliques'} condition, satisfied by chordal graphs, to the weaker condition \emph{`bags have bounded diameter'}. The bags are not required to be connected: the diameter is measured in terms of the distance between the vertices in the entire graph.

We provide a brief explanation of our method and highlight the new challenges compared to the approaches in \cite{mathieu2013graph} and \cite{KannanMZ14}. We also start by finding a vertex $v$ that lies on many shortest paths (with high probability), although we give a new approach for doing so.
In fact, our overall algorithm is more efficient than that of  \cite{mathieu2013graph,KannanMZ14} by a $(\log n)$-factor, and this is the place where we gain this improvement.
We then show that for such a vertex $v$, the set $S=N^{\leq 3k/2}[v]$ of vertices at distance at most $3k/2$ is a good separator, for $k$ the treelength of the input graph. We compute the components of $G\setminus S$ to check that indeed we found a good separator and then recursively reconstruct the components until we reach a sufficiently small vertex set on which a brute-force approach can be applied. It is key to our recursive approach, and requires non-trivial proofs, that we can add a small boundary set and still preserve all the relevant distances for a component. This problem is easily avoided in \cite{mathieu2013graph,KannanMZ14} where separators are cliques, but is more delicate to handle in our case. For this, we amongst others obtain a structural property of graphs with bounded treelength. This property is stated in the following lemma, which may be of independent interest. 
\begin{restatable}{lemma}{pathstayclose}
\label{lem:paths_stay_close}
Let $G$ be a graph of treelength at most $k\geq 1$ and $A \subseteq V(G)$. If $G[A]$ is connected then every shortest path in $G$ between two vertices $a,b \in A$ is contained in $N^{\leq 3k/2}[A]$. 
\end{restatable}

\paragraph{Roadmap} In \cref{sec:prel}, we set up our notation and give the relevant definitions. In \cref{sec:tl}, we give our algorithm to reconstruct bounded treelength graph with a proof of correctness and complexity. In \cref{sec:concl} we conclude with some open problems.

\section{Preliminaries}
\label{sec:prel}
In this paper, all graphs are simple, undirected and connected except when stated otherwise. All logarithms in this paper are base 2, unless mentioned otherwise. For $a \leq b$ two integers, let $[a,b]$ denote the set of all integers $x$ satisfying $a \leq x \leq b$. We short-cut $[a]=[1,a]$.

For a graph $G$ and two vertices $a,b \in V(G)$, we denote by $d_G(a,b)$ the length of a shortest path between $a$ and $b$. For $G = (V,E)$, $A \subseteq V$ and $i \in \NN$, we denote by $N^{\leq i}_G[A] = \{v \in V \mid  \exists a \in A, d_G(v,a) \leq i\}$. We may omit the superscript when $i=1$. We write $N_G(A)=N_G[A]\setminus A$ and use the shortcuts $N_G[u],N_G(u)$ for  $N_G[\{u\}],N_G(\{u\})$ when $u$ is a single vertex.  We may omit the subscript when the graph is clear from the context. 

\paragraph{Distance queries} We denote by $\qu_G(u,v)$ the call to an oracle that answers $d_G(u,v)$, the distance between $u$ and $v$ in a graph $G$. For $A,B$ two sets of vertices, we denote by $\qu_G(A,B)$ the $|A|\cdot|B|$ calls to an oracle, answering the list of distances $d_G(a,b)$ for all $a \in A$ and all $b \in B$. We may abuse notation and write $\qu_G(u,A)$ for $\qu_G(\{u\},A)$ and may omit $G$ when the graph is clear from the context.

For a graph class $\mathcal{G}$ of connected graphs, we say an algorithm reconstructs the graphs in the class if for every graph $G\in \mathcal{G}$ the distance profile obtained from the queries is not compatible with any other graph from $\mathcal{G}$. The \textit{query complexity} is the maximum number of queries that the algorithm takes on an input graph from $\mathcal{G}$, where the queries are adaptive. For a randomised algorithm, the query complexity is given by the expected number of queries (with respect to the randomness in the algorithm).

\paragraph{Tree decomposition} A \textit{tree decomposition} of a graph $G$ is a tuple $(T,(B_t)_{t \in V(T)})$ where $T$ is a tree and $B_t$ is a subset of $V(G)$ for every $t \in V(T)$, for which the following conditions hold.
    \begin{itemize}
        \item For every $v \in V(G)$, the set $\{t \in V(T) \mid v \in B_t\}$ is non-empty and induces a subtree of $T$.
        \item For every $uv \in E(G)$,
        there exists a $t \in V(T)$ such that $\{u,v\} \subseteq B_t$.
    \end{itemize}
This notion was introduced by  \cite{robertson1986graph}.

\paragraph{Treelength} The \textit{treelength} of a graph $G$ (denoted $\tl(G)$) is the minimal integer $k$ for which there exists a tree decomposition $(T,(B_t)_{t\in V(T)})$ of $G$ such that $d(u,v) \leq k$ for every pair of vertices $u,v$ that share a bag (i.e. $u,v\in B_t$ for some $t\in V(T)$). We refer the reader to \cite{dourisboure2007tree} for a detailed overview of the class of bounded treelength graphs.

\paragraph{Balanced separators} For $\beta\in (0,1)$, a \textit{$\beta$-balanced separator} of a graph $G = (V,E)$ for a vertex set $A\subseteq V$ is a set $S$ of vertices such that the connected components of $G[A \setminus S]$ are of size at most $\beta |A|$. 

One nice property of tree decompositions is that they yield $\frac12$-balanced separators.
\begin{lemma}[\cite{robertson1986graph}]
    \label{lem:td_sep}
    Let $G$ be a graph, $A\subseteq V(G)$ and $(T,(B_t)_{t \in V(T)})$ a tree decomposition of $G$. Then there exists $t \in V(T)$ such that $B_t$ is a $\frac12$-balanced separator of $A$ in $G$. 
\end{lemma}

\section{Randomised algorithm for bounded treelength}
\label{sec:tl}
We give the complete proof of \cref{thm:randtl} in this section. 
\randtl*

Given a tree decomposition $(T, (B_t)_{t \in V(T)})$ of a graph $G$ and a set $X$ of vertices of $G$, we denote by $T_X$ the subtree of $T$ induced by the set of vertices $t \in V(T)$ such that $B_t$ contains at least one vertex of $X$. Given $v\in V(G)$, we may abuse notation and use $T_v$ as the subtree $T_{\{v\}}$.
We first prove the following useful property of graphs of bounded treelength. 
\pathstayclose*

\begin{proof}
    Consider a tree decomposition $(T,(B_t)_{t\in V(T)})$ of $G$ such that any two vertices $u,v$ in the same bag satisfy $d(u,v) \leq k$. If two vertices $a,b\in A$  share a bag, then $d(a,b) \leq k$ and the claim holds for this pair. 
    
    Otherwise, $T_a$ and $T_b$ are disjoint subtrees of $T$ and we can consider the unique path $P$ in $T$ between $T_a$ and $T_b$, with internal nodes taken from $V(T)\setminus V(T_a)\cup V(T_b)$. We also consider  a shortest path $Q:=\{q_1,q_2,\ldots, ,q_m\}$ between $a$ and $b$ in $G$ with $q_1 = a, q_m = b$ and $
    q_i q_{i+1} \in E(G)$ for all $i < m$. 
    Since $A$ is supposed connected, $T_A$ is well-defined and is a subtree of $T$. Moreover $T_A$ contains both $T_a$ and $T_b$. Because $T_A$ is a tree, it must then contains $P$ as the unique path between $T_a$ and $T_b$.
    Suppose now, towards a contradiction, that there is some vertex $z \in Q$ such that $z \notin N^{\leq 3k/2}[A]$. Note that $T_z$ can not have common vertices with $P$ because we assumed $d(z,A) > k$ using the previous remark and the fact that vertices that share a bag are at distance at most $k$. We can then consider the vertex $t\in P$ such that $\{t\}$ separates $P \setminus \{t\}$ from $T_z$ in $T$. The shortest path $Q$ must go through $B_t$ twice: once to go from $a$ to $z$ and once to go from $z$ to $b$. 
    
    Let $i<\ell<j$ be given such that $q_i,q_j \in B_t$ and $q_\ell = z$. Since $Q$ is a shortest path in $G$,  $d(q_i,z) + d(z,q_j)=d(q_i,q_j)$. Moreover, $d(q_i,q_j) \leq k$ because $q_i$ and $q_j$ share a bag. By the pigeonhole principle, we deduce that either $d(p_i,z) \leq k/2 $ or $d(p_j,z) \leq  k/2 $. Suppose that $d(p_i,z) \leq k/2$. Remember that $t \in P$ thus $B_t$ contains an element of $A$ as $G[A]$ is connected. It follows that $d(p_i,A) \leq k$ thus $d(z,A) \leq d(z,p_i) + d(p_i,A) \leq  3k/2$, which is a contradiction. The other case follows by a similar argument.
\end{proof}
We now sketch the proof of \cref{thm:randtl}. The skeleton of the proof is inspired by \cite{KannanMZ14}: we find a balanced  separator $S$, compute the partition of $G \setminus S$ into connected components, and reconstruct each component recursively. In order to find this separator, we use a notion of \emph{betweenness} that roughly models the number of shortest paths a vertex is on.

We prove four claims. The first one ensures that in graphs of bounded treelength, the \emph{betweenness} is always at least a constant. Then, the next three claims are building on each other to form an algorithm that computes the partition of $G \setminus S$ into connected components of roughly the same size.
\begin{itemize}
    \item \cref{cl:find_z} is a randomised procedure for finding a vertex $z$ with high betweenness (using few queries and with constant success probability).
    \item \cref{cl:find_balanced_separtor_tl} shows $S=N^{\leq 3k/2}[z]$ is a good balanced separator if $z$ has high betweenness.
    \item \cref{cl:find_cc} computes the partition of $G \setminus S$ into connected components. Note that, once you computed the partition, you can check if the preceding algorithms have been successful. If not, we can call again \cref{cl:find_balanced_separtor_tl} until we are successful, yielding a correct output with a small number of queries in expectation.
\end{itemize}

\begin{proof}[Proof of Theorem \ref{thm:randtl}]
Let $G$ be a connected $n$-vertex graph of maximum degree at most $\Delta$ and let $(T,(B_t)_{t\in V(T)})$ be a tree decomposition of $G$ such that $d(u,v)\leq k$ for all $u,v\in V(G)$ that share a bag in $T$.

We initialize $A=V(G)$, $n_A = |A|$ and $R^i=\emptyset$ for $i \in [1,3k]$. For any $j \in \mathbb{R}^+$ we abbreviate  $R^{\leq j}=\cup_{i\leq j} R^{i}$. Lastly, let $r = |R^{\leq 3k}$|.
We will maintain throughout the following properties:
\begin{enumerate}
    \item $G[A]$ is a connected induced subgraph of $G$.
    \item $R^i$ consists of the vertices in $G$ that are at distance exactly $i$ from $A$.
    \item Both $A$ and $R^i$ for all $i$ are known by the algorithm.
\end{enumerate}
In particular, we know which vertices are in sets such as  $R^{\leq 3k/2} = N^{\leq 3k/2}[A]$ and by Lemma \ref{lem:paths_stay_close} we also obtain the following crucial property.
\begin{enumerate}
    \item[4.] For $a,b\in A$, any shortest path between $a$ and $b$ only uses vertices from $A \cup R^{\leq 3k/2}$. 
\end{enumerate}
The main idea of the algorithm is to find a balanced separator $S$ and compute the partition of $G[A\setminus S]$ into connected components, then call the algorithm recursively on each components. As soon as $n_A$ has become sufficiently small, we will reconstruct $G[A]$ by `brute-force queries'.

In order to find the separator $S$, we use the following notion. For $G$ a graph, a subset $A \subseteq V(G)$ and a vertex $v\in V(G)$, the betweenness $p_v^G(A)$ is the fraction of pairs of vertices $\{a,b\}\subseteq A$ such that $v$ is on some shortest path in $G$ between $a$ and $b$. We first prove that there is always some vertex $v\in A\cup R^{\leq k}$ (a set known to our algorithm) for which $p_v(A)$ is large.
\begin{claim}
\label{cl:lowerbound_p}
We have
    $p:=\max\limits_{v\in A \cup R^{\leq k}}p_v^{G}(A)\geq \frac1{2(\Delta^k+1)}$.
\end{claim}
\begin{proof}
Our original tree decomposition also restricts to a tree decomposition for $G[A]$, so Lemma \ref{lem:td_sep} shows that there exists a bag $B$ of $T$ such that 
    $B$ is a $\frac12$-balanced separator of $G[A]$. Note that $G[A]$ is connected, so there exists some $a\in A\cap B$. As $T$ is a witness of $G$ being of bounded treelength, the distance between any two vertices of $B$ is at most $k$. In particular, $B\subseteq N^{\leq k}[a] \subseteq A \cup R^{\leq k}$, and $|B|\leq \Delta^k+1$ since $G$ has maximum degree $\Delta$. 
    Moreover, since $B$ is a $\frac12$-balanced separator of $G[A]$, for at least half of the pairs $\{u,v\} \subseteq A$, the shortest path between $u$ and $v$ goes through $B$. Using the pigeonhole principle, there exists a $v \in B$ such that $p^G_v(A) \geq \frac1{2(\Delta^k+1)}$.
\end{proof}

The next three claims are building on each other to find a balanced separator $S$.
In the first one, we argue that we can find, using few queries, a vertex with high betweenness.
\begin{claim}
\label{cl:find_z}
There is a randomised algorithm that finds $z\in N^{\leq 3k/2}[A]$ with $p_z^G(A)\geq p/2$ with probability at least 2/3 using $O(p^{-1}(n_A+r)\log(n_A+r))$
distance queries in $G$.
\end{claim}
\begin{proof}
To simplify notation, we omit $G$ and $A$ from $p^G_v(A)$ and only write $p_v$. We first sample uniformly and independently (with replacement) pairs of vertices $\{u_i,v_i\} \subseteq A$ for $i \in [C \log (n_A+r)]$ where $C\leq \frac1{2p}+1$ is defined later. Then, we ask $\qu(u_i,N^{\leq 3k/2}[A])$ and $\qu(v_i,N^{\leq 3k/2}[A])$.

We write\[
P_i = \{x \in N^{\leq 3k/2}[A] \mid  d(u_i,x) + d(x,v_i) = d(u_i,v_i)\}
\]
for the set of vertices that are on a shortest path between $u_i$ and $v_i$. Note that \cref{lem:paths_stay_close} implies that $P_i$ contains all vertices of $V(G)$ on a shortest path from $u_i$ to $v_i$. From the queries done above we can compute $P_i$ for all $i\in [C\log (n_A+r)]$.
For each vertex $v \in N^{\leq 3k/2}[A]$, we denote by $\Tilde{p}_v$ an estimate of $p_v$ defined by $\Tilde{p}_v = |\{i\in [C\log(n_A+r)] : v \in P_i\}| /({C \log (n_A+r)})$.
The algorithm outputs $z$ such that $z = \arg\max_{v \in N^{\leq 3k/2}[A]} \Tilde{p}_v$.

The query complexity of this algorithm is $2C\log(n_A+r)|N^{\leq 3k/2}[A]|=O_{k,\Delta}(n_A\log(n_A+r))$

We now justify the correctness of this algorithm and give $C$. Let $y = \arg\max_{w \in N^{\leq 3k/2}[A]} p_w$. We need to show that $p_z\leq \frac{p_y}2$ has probability at most $\frac13$.
Let $u$ be a vertex chosen uniformly at random among the set of vertices $w\in N^{\leq 3k/2}[A]$ with $p_w\leq p_y/2$. A simple union bound  implies that it is sufficient to show that $\mathbb{P} [ \Tilde{p}_{y} \leq \Tilde{p}_u] < 1/(3n_A+3r)$. Indeed, this implies that the probability that a vertex $w$ with $p_w\leq p_y/2$ is a better candidate for $z$ than $y$, is at most $1/3$. Note that the elements of $\{\Tilde{p}_w \mid w \in N^{\leq 3k/2}[A]\}$ (and thereby $z$) are random variables depending on the pairs of vertices sampled at the start, and that the elements of $\{p_w \mid w\in N^{\leq 3k/2}[A]\}$ are fixed.

We denote by $A_i$ the event $\{u \in P_i\}$ and by $B_i$ the event $\{y \in P_i\}$. The events $(A_i)_i$ are independent, since each pair $\{u_i,v_i\}$ has been sampled uniformly at random and independently. By definition, $\Pr[A_i] = p_u \leq p_{y}/2$ and $\Pr[B_i] = p_{y}$. Thus, the random variable $X_i$ defined by $X_i = \mathbbm{1}_{A_i} - \mathbbm{1}_{B_i}$ has expectation $\Ex[X_i] \leq -p_{y}/2$. Therefore, applying Hoeffding's inequality \cite{hoeffding1994probability}, we obtain
\begin{align*}
\Pr\left[\sum_{i=1}^{C\log (n_A+r)} X_i \geq 0 \right] &\leq 2 \exp({-\frac{2(C\log(n_A+r)p_y/2)^2}{4\log(n_A+r)}}).\\
\end{align*}
By choosing $\frac1{2p}+1\geq C \geq \frac1{2p_y} = \frac1{2p}$ such that $C\log(n_A+r)$ is an integer, we conclude that
\[
\Pr[ \Tilde{p}_{y} \leq \Tilde{p}_u]= 
\Pr\left[\sum_{i=1}^{C\log(n_A+r)} X_i \geq 0 \right] \leq 2 \exp({-2\log(n_A+r)}) \leq 1/(3n_A+3r)
\]
for $n_A\geq 6$. This completes the proof.
\end{proof}
Let $z$ be a vertex with high betweenness as in the claim above. We now argue that $N^{3k/2}[z]$ is an $\alpha$-balanced separator for some constant $\alpha$ depending only on $\Delta$ and $k$.

\begin{claim}
\label{cl:find_balanced_separtor_tl}
Let $\alpha = \sqrt{1 - \frac1{4(\Delta^k +1)}}$. If $z \in N^{\leq 3k/2}[A]$ satisfies $p_z^G(A)\geq p/2$, then $S:= N^{\leq 3k/2}[z]$ is an $\alpha$-balanced separator for $A$. 
\end{claim}
\begin{proof}
    Suppose towards contradiction that $S$ is not an $\alpha$-balanced separator. Thus there is a connected component $C$ of $G[ V(G) \setminus S]$ with $|C\cap A|> \alpha n_A$. By definition of $S$, $d(z,C) > 3k/2$ which implies by \cref{lem:paths_stay_close} that for any pair of vertices in $C$, no shortest path between these two vertices goes through $z$. In particular, this holds for pairs of vertices in $C \cap A$. Therefore, 
    \[
    p_z^G(A)\leq \frac{(n_A^2-|C\cap A|^2)}{n_A^2} < 1 - \alpha^2
 = 1 - (1 - \frac1{4(\Delta^k +1)}) = \frac1{4(\Delta^k +1)} \leq p/2
 \]
 using Claim \ref{cl:lowerbound_p} for the last step, contradicting our assumption that $p_z^G(A)\geq p/2$.
\end{proof}
We apply \cref{cl:find_z} to find $z\in N^{\leq 3k/2}$, where $p_z^G(A)\geq p/2$ with  probability at least $2/3$ (using also \cref{cl:lowerbound_p}). We compute $S =N^{\leq 3k/2}[z]$ using $O_{k,\Delta}(n_A + r)$ distance queries; this can be done since $S \subseteq A \cup R^{\leq 3k}$ so the algorithm only needs to consider $n_A+r$ vertices when searching for neighbours. 

The set $S$ is an $\alpha$-balanced separator with probability at least $2/3$ by \cref{cl:find_balanced_separtor_tl}. In particular, the algorithm does not know yet at this point if it is indeed a good separator or not. It will be able to determine this after computing the partition of $G[A\setminus S]$ into connected components.

The following claim uses mostly the same algorithm as 
\cite[Alg. 6]{KannanMZ14}, and the proof is analogous. As we are using this algorithm in a slightly different setting, we still give a complete proof of the lemma.

\begin{claim}
\label{cl:find_cc}
    There is a deterministic algorithm that given a set $S \subseteq A$, computes the partition of $A\setminus S$ into connected components of $G[A \setminus S]$ using at most $n_A \cdot \Delta (r + |S|)$ distance queries.
\end{claim}

\begin{proof}
By assumption, $R^1$ is the set of vertices at distance exactly 1 from $A$ in $G$. Since $A$ is connected, it is a connected component of $G[V(G)\setminus R^1]$. Therefore, the connected components of $G[A\setminus S]$ are exactly the connected components of $G[V(G) \setminus (R^{1} \cup S)]$ containing an element of $A$. We denote by $B$ the open neighbourhood of $S \cup R^1$ in $A$, that is, $B = (N[S\cup R^1] \cap A) \setminus (S \cup R^1)$.
    We use the following algorithm.
    \begin{itemize}
        \item We ask $\qu(A, S \cup R^1)$ in order to deduce $N[S \cup R^1] \cap A$, and then we ask $\qu(A,N[S \cup R^1] \cap A)$.
        \item We compute $D_b = \{ v \in A \cap S \mid d(v,b) \leq d(v, S \cup R^1)\}$ for $b \in B$, the set of vertices in $A\cap S$ which have a shortest path to $b$ that does not visit a vertex of $S\cup R^1$. 
        \item  Let $\mathcal{D} = \{ D_s \mid s \in B \}$. While there are two distinct elements $D_1,D_2 \in \mathcal{D}$ such that $D_1 \cap D_2 \neq \emptyset$, merge them in $\mathcal{D}$, that is, update $\mathcal{D} \leftarrow (\mathcal{D} \setminus \{D_1,D_2\}) \cup \{ D_1\cup D_2 \}$. We output $\mathcal{D}$.
    \end{itemize}
Note that any vertex $a\in A\setminus S$, is not in $S\cup R_1$, so will be in $D_s$ for at least one $s\in B$ (possibly $s=a$) before we do the last step of the algorithm. The last step ensures that the output is indeed a partition of $A$. 

We first argue that $\mathcal{D}$, as outputted by the algorithm above, is an over-approximation of the connected component partition of $G[A\setminus S]$ (that is, for any connected component $C$ of $G[A\setminus S]$, there exists $D \in \mathcal{D}$ such that $C \subseteq D$). It suffices to prove that for any edge $ab \in E(G[A\setminus S])$ there exists $D \in \mathcal{D}$ such that $\{a,b\} \subseteq D$. Suppose without loss of generality that $d(a,S \cup R^1) \leq d(b,S \cup R^1)$. Moreover let $s \in B$ such that $ d(a,s) = d(a,S\cup R^1) - 1$ and thus $a \in D_s$. Now $d(b,s) \leq d(a,s) + 1 \leq d(b,S \cup R^1)$ thus $b \in D_s$.  
    
We now argue that  $\mathcal{D}$ is an under-approximation too, by showing that $G[D\setminus S]$ is connected for all $D \in \mathcal{D}$. We first show this for the initial sets $D_s$ with $s\in N[S\cup R^1]\cap A$. Let $s \in B$. For any $v \in D_s$, by definition, $d(v,s) \leq d(v,S \cup R^1)$, thus there is a shortest path $P$ between $v$ and $s$ not using vertices of $S \cup R^1$. Moreover $s \in A$ and $A$ is separated from $V(G) \setminus A$ by $R^1$, therefore $P$ is contained in $A \setminus S$. This shows that $v$ is in the same connected component of $G[A\setminus S]$ as $s$. To see that $G[D]$ remains connected for all $D\in \mathcal{D}$ throughout the algorithm, note that when the algorithm merges two sets $D_1,D_2 \in \mathcal{D}$, they need to share a vertices, thus if both $G[D_1]$ and $G[D_2]$ are connected then $G[D_1 \cup D_2]$ is also connected. 

Remember that $|S| \leq \Delta^{3k/2} + 1 = O_{k,\Delta}(1)$ and that the bounded degree condition implies $|N[S \cup R^1]| \leq \Delta \cdot |S \cup R^1|$. This allow us to conclude that the query complexity is bounded by
\[
|A|\cdot|N[S\cup R^1]| \leq n_A\cdot \Delta |S \cup R^1| \leq n_A \cdot \Delta (r + |S|).\qedhere
\]
\end{proof}
We apply the algorithm from \cref{cl:find_cc} with the separator $S$ computed by \cref{cl:find_balanced_separtor_tl}. Knowing the partition, the algorithm can check if $S$ is indeed $\alpha$-balanced. If not, the algorithm repeats \cref{cl:find_cc} and computes a new potential separator. 
An single iteration succeeds with probability at least $2/3$ and each iteration is independent from the others, so the expected number of repetitions is $3/2$. 

We ask $\qu(S \cup R^{\leq 3k}, A)$. For each connected component $\Tilde{A}$ of $G[A \setminus S]$, we will reconstruct $G[\Tilde{A}]$ and then we will describe how to reconstruct $G[A]$.
If $|\Tilde{A}|\leq \log(n)$, then we ask $\qu(\Tilde{A},\Tilde{A})$ to reconstruct $G[\Tilde{A}]$. Otherwise, we will place a recursive call on $\Tilde{A}$, after guaranteeing that our desired properties mentioned at the start are again satisfied.
By definition, $G[\Tilde{A}]$ is connected. 
So we know property 1 holds when $A$ is replaced by $\Tilde{A}$.

To ensure properties 2 and 3 are also satisfied for the recursive call, we reconstruct $\Tilde{R}^i$, the set of vertices at distance exactly $i$ from $\Tilde{A}$.
As $S\cup R^1$ separates $\Tilde{A}$ from other component of $G[A\setminus S]$, for any other 
connected component $D$ of $G[A\setminus S]$ and for any $v\in D$, we have:
$$
d(\Tilde{A},v) = \min\limits_{s \in S \cup R^1} d(\Tilde{A},s) + d(s, v).
$$ 
Therefore we can compute $d(\Tilde{A},v)$ from the query results of $\qu(S \cup R^{\leq 3k}, A)$ for all $v \in A \cup R^{\leq 3k}$. This is enough to deduce $\Tilde{R}^i$ for any $i \leq 3k$ because $\Tilde{A} \subseteq A$ and thus $\Tilde{R}^{i} \subseteq A \cup R^{i}$. 

After we have (recursively) reconstructed $G[\Tilde{A}]$ for each connected component $\Tilde{A}$ of $G[A\setminus S]$, we reconstruct $G[A]$ by using that we already know all the distance between all pairs $(a,s)$ with $a\in A$ and $s\in S$. In particular, as we already asked $\qu(S\cup R^{\leq 3k},A)$ earlier in the algorithm, we know $G[S\cap A]$ and also how to `glue' the components to this (namely, by adding edges between vertices at distance 1).

By \cref{cl:find_balanced_separtor_tl}, each recursive call reduces the size of the set $A$ under consideration by a multiplicative factor of $\alpha$. Therefore, the recursion depth is bounded by $O_{\Delta,k}(\log n)$ and the algorithm will terminate.

We argued above that the algorithm correctly reconstructs the graph. It remains to analyse the query complexity.

We analyse the query complexity via the recursion tree, where we generate a child for a vertex when it places a recursive call. We can associate to each vertex $v$ of the recursion tree $T_R$, a subset $A_v\subseteq V(G)$ for which the algorithm is trying to reconstruct $G[A_v]$. The subsets associated to the children of a node $v$ are disjoint, since each corresponds to a connected component of $A_v\setminus S_v$ for some subset $S_v\subseteq V(G)$ that is an $\alpha$-balanced separator. In particular, the subsets associated to the leafs are disjoint.

In a leaf node $v$, the algorithm performs $|A_v|^2$ queries to reconstruct $G[A_v]$, where  $|A_v|\leq \log(n)$. If we enumerate the sizes $A_v$ for the leafs $v$ of the recursion tree as $a_1,\dots,a_\ell$, then 
$\sum_{i=1}^\ell a_i^2\leq \ell \log(n)^2\leq n \log(n)^2$, where we use that we have at most $n$ leafs since the corresponding subsets are disjoint.

Since there are at most $n$ leafs, and the recursion depth is $O_{k,\Delta}(\log n)$, there are $O_{k,\Delta}(n \log n)$ internal nodes.
Let $v$ be an internal node and let $n_A$ and $r$ denote the sizes of the corresponding subsets $A=A_v$ and $R^{\leq 3k}$. The algorithm makes the following queries:
\begin{itemize}
    \item Finding $z$ takes $O_{k,\Delta}(n_A\log(n_A+r))$ queries in Claim \ref{cl:find_z}. 
    \item $O_{k,\Delta}(n_Ar)$ queries to compute $S$ from $z$ and to find the connected components of $A\setminus S$ in Claim \ref{cl:find_cc}. This step and the previous step are repeated a constant number of times (in expectation).
    \item $O_{k,\Delta}(n_Ar)$ queries to set up the recursive calls to the children of $v$.
\end{itemize}
Since each recursive call increases the size of $R^{\leq 3k}$ by at most an additive constant smaller than $(\Delta+1)^{9k/2}$ (recall that $\Tilde{R}^{\leq 3k}\subseteq R^{\leq 3k}\cup N^{\leq 9k/2}[z]$), and the recursion depth is $O_{k,\Delta}(\log n)$, it follows from an inductive argument that $r= O_{k,\Delta}(\log n)$. So the number of queries listed above is $O_{k,\Delta}(n_A\log n)$.

To compute the total query complexity of internal nodes, we use the fact that for two nodes $v$ and $v'$ at the same recursion depth we have that $A_v \cap A_{v'} = \emptyset$. Therefore, by adding contribution layer by layer in the recursion tree we get a query complexity of $O_{k,\Delta}(n\log n)$ for any fixed layer, and the total number of queries performed sum up to:
\[
n\log^2 n +O_{k,\Delta}(\log n) O_{k,\Delta}(n\log n)=O_{k,\Delta}(n\log^2n). \qedhere
\]

We did not try to optimise the dependence in $k$ and $\Delta$ hidden in the $O_{k,\Delta}$ notation throughout the proof of \cref{thm:randtl}. Expanding all $O_{k,\Delta}$ notations in the proof implies that our algorithm uses $\Delta^{O(k)} n \log^2 n$ queries. It would be interesting to reduce this dependence to a polynomial in $\Delta$ and $k$.

\end{proof}

\section{Conclusion}
\label{sec:concl}
In this paper, we shed further light on what graph structures allow efficient distance query reconstruction.  We expect that the true deterministic and randomised query complexity of graphs of bounded bounded treelength and bounded maximum degree is $\Theta(n\log n)$, matching the lower bound which already holds for trees from \cite{bastide2023optimal}. 

It seems natural that having small balanced separators helps with obtaining a quasi-linear query complexity. We show this is indeed the case when some additional structure on the separator is given (namely, vertices being `close'). 
A possible next step would be to see if this additional structure can be removed.
\begin{problem}
Does there exist a randomised algorithm that reconstructs an $n$-vertex graph of maximum degree $\Delta$ and treewidth $k$ using $\Tilde{O}_{\Delta,k}(n)$ queries in expectation?
\end{problem}
Some parts of the algorithm still work, such as checking whether a given set $S$ is a balanced separator (via Claim \ref{cl:find_cc}). When trying to recursively reconstruct one of the components, it is important to `keep enough information about the distances'. In our algorithm, we can include the shortest paths between the vertices in the separator; this is the main purpose of the `boundary sets' $R^i$ and why we carefully chose the domain for $z$ in Claim \ref{cl:find_z}. The possibility to do this is almost the definition of bounded treelength. Therefore, we believe that a new approach would be needed to produce a good candidate for a balanced separator in the general case.

Finally, we remark that it may very well be that techniques building on separators are needed as part of a potential quasi-linear algorithm for reconstructing graph classes that do not directly guarantee the existence of such separators. Indeed, there are approaches that actually do not work well even on trees, yet are good at handling certain graphs without small balanced separators, and perhaps a combination of both types of methods will be needed to handle the class of all bounded degree graphs. For example, the approach taken by \cite{mathieu2021simple} is to ask all queries to a randomly selected set of vertices. On some graph classes (such as random regular graphs, which do not have small balanced separators), this already forces most of the non-edges with high probability and so the remaining pairs can be queried directly. But in order to beat the best-known upper bound for general graphs of bounded degree (of $\widetilde{O}_\Delta(n^{3/2})$ from \cite{mathieu2013graph}), such an approach cannot be applied directly, even for trees. Indeed, for a complete binary tree on $n$ vertices, the distances to any set $S$ with at most $\frac1{100}\sqrt{n}$ vertices, no matter how cleverly chosen, leave many pairs of distances undetermined. In fact, there are approximately $\sqrt{n}$ vertices at height $\lfloor \tfrac12\log n\rfloor$ in this tree, and $S$ will miss the `trees below' most of those $\sqrt{n}$ vertices entirely. This means that there are still $\Omega(n^{3/2})$ pairs $u,v$ that form a non-edge, yet have the same distance to all vertices in $S$. This means that even for the class of all bounded degree graphs, there may need to be a part of the algorithm which exploits the structure of `nice' separators, when they exist.

\bibliographystyle{plainurl}
\bibliography{biblio}

\end{document}